\documentclass[12pt]{article}
\usepackage{graphicx}
\usepackage{amsmath}
\usepackage{amsfonts,amssymb}

\usepackage[T1]{fontenc}

\usepackage[latin1]{inputenc}

\newcommand{\newc}{\newcommand}
\newc{\beq}{\begin{equation}}
\newc{\eeq}{\end{equation}}
\newc{\bea}{\begin{array}}
\newc{\eea}{\end{array}}
\newcommand{\ben}{\begin{eqnarray}}
\newcommand{\een}{\end{eqnarray}}
\newc{\ra}{\rightarrow}
\newc{\bfx}{{\bf x}}
\newc{\bfV}{{\bf V}}
\newc{\cO}{{\cal O}}
\newc{\bfv}{{\bf v}}
\newc{\bfu}{{\bf u}}
\newc{\bfp}{{\bf p}}
\newc{\ve}{{\varepsilon}}
\newc{\Psibar}{\overline\Psi}
\newc{\w}{{\bf w}}
\newc{\E}{{\mathbf{E}}}
\newc{\EE}{{\mathcal E}}
\newc{\bfn}{{\mathbf\nabla}}
\newc{\la}{{\cal L}}
\newc{\tla}{{\tilde{\cal L}}}
\newc{\bp}{{\bf p}}
\newc{\ho}{\hookrightarrow }
\newc{\bP}{{\bf P}}
\newc{\pd}{{\partial}}
\newc{\piv}{{\partial_4}}
\newc{\pv}{{\partial_5}}
\newc{\bJ}{{\bf J}}
\newc{\bze}{{\mathbf 0}}
\newc{\bK}{{\bf K}}
\newc{\tphi}{{\tilde\phi}}
\newc{\tF}{{\tilde F}}
\newc{\tD}{{\tilde D}}
\newc{\tJ}{{\tilde J}}
\newc{\tj}{{\tilde j}}
\newc{\bD}{{\bf D}}
\newc{\tvphi}{{\tilde\varphi}}
\newc{\trho}{{\tilde\rho}}
\newc{\ttheta}{{\tilde\theta}}
\newc{\tpsi}{{\tilde\psi}}
\newc{\tu}{{\tilde u}}
\newc{\cD}{{\cal D}}
\newc{\tPhi}{{\tilde\Phi}}
\newc{\tPsi}{{\tilde\Psi}}
\newc{\tA}{{\tilde A}}
\newc{\talpha}{{\tilde\alpha}}
\newc{\tbeta}{{\tilde\beta}}
\newc{\bA}{{\mathbf A}}
\newc{\bB}{{\bf B}}
\newc{\br}{{\bf r}}
\newc{\sig}{{\mathbf\sigma}}
\newc{\eg}{{\rm e.g.\ }}
\newc{\ie}{{\rm i.e.\ }}
\newcommand{\bey}{\begin{eqnarray}}
\newcommand{\pslash}{\not{\hbox{\kern-2.3pt $p$}}}
\newcommand{\pdslash}{\not{\hbox{\kern-2pt $\partial$}}}
\newcommand{\eey}{\end{eqnarray}}
\newtheorem{theorem}{Theorem}
\newtheorem{lemma}{Lemma}

\newenvironment{proof}[1][Proof]{\noindent\textbf{#1.} }{\ \rule{0.5em}{0.5em}}

\begin{document}

\hyphenation{ope-ra-tor}
\hyphenation{qua-si-tri-an-gu-lar}
\hyphenation{cha-ra-cte-ri-zing}
\hyphenation{cha-ra-cters}
\hyphenation{a-sym-pto-tic}
\hyphenation{or-tho-go-na-li-ty}

\begin{titlepage}
\vskip 2cm
\begin{center}
{\Large  Factor groups, semidirect product and quantum chemistry.
\footnote{{\tt matrindade@uneb.br}}}
 \vskip 10pt
{ M. A. S. Trindade \\}
\vskip 5pt
{\sl Departamento de Ciências Exatas e da Terra, Universidade do Estado da Bahia\\
Rodovia Alagoinhas/Salvador, BR 110, Km 03, 48040-210 - Alagoinhas, Bahia, Brazil\\}
\vskip 2pt
\end{center}

\begin{abstract}

In this paper we prove some general theorems about representations of finite groups arising from the inner semidirect product of groups. We show how these results can be used for standard applications of group theory in quantum chemistry through the orthogonality relations for the characters of irreducible representations. In this context, conditions for transitions between energy levels, projection operators and basis functions were determined. This approach applies to composite systems and it is illustrated by the dihedral group related to glycolate oxidase enzyme.
\end{abstract}

\bigskip


\vskip 3pt

\end{titlepage}


\newpage

\setcounter{footnote}{0} \setcounter{page}{1} \setcounter{section}{0} %
\setcounter{subsection}{0} \setcounter{subsubsection}{0}
\section{Introduction}
   The concept of symmetry is ubiquitous in quantum mechanics. For example, the elementary particles can be classified using the irreducible representations of continuous symmetry groups. This approach begins at a seminal paper of Wigner \cite{Wigner1} with systematic study of unitary representations of Poincaré group in relativistic quantum mechanics. Subsequently, a similar study was carried out in the scope of non-relativistic quantum mechanics through the invariance of the Schrödinger equation by Galilei group \cite{Bargmann, Inonu, Levy} and the theory of quarks, developed by Gell-Mann, \cite{Gell1,Gell2} may be understood with the flavour symmetries. Evidently, there are many other important findings related to applications of Lie groups and we do not pretend to cover them all.

   In the context of finite groups, representation theory plays a fundamental role in quantum chemistry \cite{Bishop, Cotton, Tinkham}. The quantum numbers are indices characterizing irreducible representations of finite groups \cite{Weyl} and the numbers and kinds of energy levels are determined by symmetry of molecule. As stressed by Wigner \cite{Wigner2}, the recognition that almost all rules of spectroscopy from the symmmetry of the systems is a remarkable result. In fact, group theoretical techniques are important in determining the rules selection for optical process such as infrared and Raman activity. Transitions of lower symmetry may lead to mode splittings. These mode splittings and the changes in the infrared and Raman spectra can be predict using group theoretical techniques \cite{Dres}.

   Despite sucess in quantum chemistry, applications of  finite groups in many other fields were obtained. In the modern theory of quantum computation many formalisms for quantum error correction uses finite groups \cite{Knill, Gottesman}. Nice error bases are characterized in terms of the existence of certain characters in a group \cite{Knill} and the characterization of decoherence-free subspaces for multiple-qubits errors can be determined by one-dimensional representation of the Pauli group \cite{Lidar1, Lidar2}. It was shown how to perform universal and fault tolerant quantum computation on decoherence-free subspaces. The quantum teleportation and the quantum dense coding in a finite-dimensional Hilbert space can be formulated in terms of an irreducible unitary representation of finite group \cite{Ban}. Quantum logic gates capable of preserving quantum entanglement may be obtained of representations of braid group through quasitriangular Hopf algebras derived from a cyclic group \cite{Pinto}. Anyonic models based on finite groups have been proposed in the scope of topological quantum computation \cite{Kitaev1, Kitaev2} and topological quantum field theories associated to finite groups have been constructed \cite{Yetter, Freed}. These examples show that the group-theoretical methods remain a powerful tool in the analysis of new quantum-mechanical problems.

   In the quantum chemistry, after seminal works, new methods have also been developed. Unconventional approaches to obtain irreducible representations through regular projection matrices \cite{Blokker}, new algorithms for point group symmetries \cite{Eick, Fritzsche} and schemes for adapted symmetry functions of point groups \cite{Peng} are some interesting examples. A beautiful application of the actions of finite groups was performed by Torres \cite{Torres}. It is present a new procedure for the attainment of the number and isotropy group of the vibrational force constants for a given molecule. This approach is useful in the case of high symmetry because it does not require the use of the matrix representation generated by the set of internal coordinates.

   An interesting algebraic structure that has been investigated in physics is the semidirect product of groups. In crystallography, the space group is called symmorphic group if it is a semidirect product of its point group with its translation group. Crystals whose space groups are of semidirect product type are called simple crystals and the point group of this type of crystal leaves not only the lattice invariant but also the crystal \cite{Lomont}. There are $73$ symmorphic space groups and every other space group is isomorphic to a subgroup of one of  semidirect product space groups. The determination of irreducible representations of symmorphic space groups using the semidirect product was performed by Bradley and Kammel \cite{Bradley} based on the seminal works of Mackey \cite{Mac} and McIntosh \cite{Mc}. These results are useful in calculating the electron energy bands in crystalline solids.  In this line, Chen \cite{Chen} presented a factorization lemma for the irreducible symmetry operators of semidirect product of two abelian groups so that symmetry adapted functions and algebraic solutions for the double-valued representation of the tetrahedral group were derived. Nevertheless, many interesting works related to this topic can be found in others fields. The Poincaré group, the asymptotic group in general relativity and invariance group of electrodynamics are examples of the semidirect product applied to certain group representations \cite{Geroch}. The symmetries of time-dependent Schrödinger equations related to semidirect product were examined by Okubo \cite{Okubo} and in the quantum computation theory, efficient algorithms have been found for several groups that be written as semidirect product of abelian groups \cite{Bacon,Inui}.

   In this work, we present some general results about representations associated to the inner semidirect product of finite groups and we show how these results can be useful in a spectroscopy analysis. The structure of this paper is as follows. In section 2 we present the mathematical results as well as possible applications in quantum chemistry. In section 3, an example using a dihedral group related to glycolate oxidase enzyme is performed. Conclusions are presented in section 4.

\section{Mathematical results}
A theorem concerning to reduction of finite solvable was given by Schur \cite{Schur}. If we have a finite solvable group with a chain of subgroups
\begin{equation}
G_{1}\supset G_{2}\supset ...\supset G_{n}
\end{equation}
where $G_{n}$ is the unit element, the irreducible representations of $G_{i}$ may be obtained from the irreducible representations of $G_{i+1}$. Later developments for the space group were obtained by Seitz \cite{Seitz}. As mentioned above, a scheme for determining the irreducible representations of space groups using the semidirect product was developed by Bradley and Kammel \cite{Bradley} through induced representations and little groups.  In this section we present some theorems about representations of semidirect product of groups. Here, we use the orthogonality relations for the characters, the isomorphism theorems of groups and the correspondence theorem. Also, we discuss applications to quantum chemistry involving composite systems. The following lemma gives a characterization of the irreducible representations for a subgroup of a group that is semidirect product of two subgroups, one of which is analyzed subgroup.

\begin{lemma}
Let $G_{1}= N_{1}\rtimes H_{1},G_{2}= N_{2}\rtimes H_{2},..., G_{n}= N_{n}\rtimes H_{n}$ be finite groups. Then every irreducible representation of $H_{1}\times H_{2}\times...\times H_{n}$ corresponds to an unique irreducible representation of $G_{1}\times G_{2}\times...\times G_{n}$.
\end{lemma}
\begin{proof}
Since we have the semidirect product of groups, $H_{1}\simeq G_{1}/N_{1}, H_{2}\simeq G_{2}/N_{2},...,H_{n}\simeq G_{n}/N_{n}$, and consequently $H_{1}\times H_{2}\times...\times H_{n}\simeq G_{1}/N_{1}\times G_{2}/N_{2}\times...\times G_{n}/N_{n}\simeq (G_{1}\times G_{2}\times...\times G_{n})/(N_{1}\times N_{2}\times...\times N_{n})$ by first isomorphism theorem of groups. Then the orthogonality relations for the characters of irreducible representations of the $(G_{1}\times G_{2}\times...\times G_{n})/(N_{1}\times N_{2}\times...\times N_{n})$ must satisfy
\begin{eqnarray}
&&\sum_{S_{1}\times S_{2}\times...\times S_{n}}[\chi^{(\mu_{1}\otimes \mu_{2}\otimes...\otimes \mu_{n}) \ast }(S_{1}\times S_{2}\times...\times S_{n}) \nonumber \\
&&\chi^{(\nu_{1}\otimes \nu_{2}\otimes...\otimes \nu_{n})}(S_{1}\times S_{2}\times...\times S_{n})] \nonumber \\
&=&\frac{g_{1}g_{2}...g_{n}}{n_{1}n_{2}...n_{n}}\delta_{\mu_{1}\mu_{2}...\mu_{n}, \nu_{1}\nu_{2}...\nu_{n}},
\end{eqnarray}
where $S_{i}$ is an element of $H_{i}$ and $g_{i}=|G_{i}|,n_{i}=|N_{i}|$. As each coset $R_{i}G_{i}$, with $R_{i} \in G_{i}$, has $n_{i}$ elements and every representation of the $G_{i}/N_{i}$ corresponds to an unique representation of $H_{i}$ , we have for the last relation
\begin{eqnarray}
&&n_{1}n_{2}...n_{n} \sum_{S_{1}\times S_{2}\times...\times S_{n}}[ \chi^{(\mu_{1}\otimes \mu_{2}\otimes...\otimes \mu_{n}) \ast }(S_{1}\times S_{2}\times...\times S_{n}) \nonumber \\
&&\chi^{(\nu_{1}\otimes \nu_{2}\otimes...\otimes \nu_{n})}(S_{1}\times S_{2}\times...\times S_{n})] \nonumber \\
&=&\sum_{R_{1}\times R_{2}\times...\times R_{n}}[\chi^{(\mu_{1}\otimes \mu_{2}\otimes...\otimes \mu_{n}) \ast }(R_{1}\times R_{2}\times...\times R_{n}) \nonumber \\
&&\chi^{(\nu_{1}\otimes \nu_{2}\otimes...\otimes \nu_{n})}(R_{1}\times R_{2}\times...\times R_{n})] \nonumber \\
&=&g_{1}g_{2}...g_{n}\delta_{\mu_{1}\mu_{2}...\mu_{n}, \nu_{1}\nu_{2}...\nu_{n}} \nonumber \\
&=&g_{1}g_{2}...g_{n}\delta_{\mu_{1} \nu_{1}}\delta_{\mu_{2} \nu_{2}}...\delta_{\mu_{n} \nu_{n}}.
\end{eqnarray}
Hence we have the orthogonality relations for the characters of irreducible representations of $G_{1}\times G_{2}\times...\times G_{n}$ and so the result follows.

\end{proof}

\begin{lemma}
Suppose $\Gamma$ and $\Gamma^{'}$ reducible representations of finite groups $G_{1}\times G_{2}\times...\times G_{n}$ and $(G_{1}\times G_{2}\times...\times G_{n})/(N_{1}\times N_{2}\times...\times N_{n})$, respectively. If $\chi(R_{1}^{i_{1}})=p_{1}\chi(S_{1}),\chi(R_{2}^{i_{2}})=p_{2}\chi(S_{2}),..., \chi(R_{n}^{i})=p_{n}\chi(S_{n})$ for $R_{1}^{i_{1}}\in G_{1}, R_{2}^{i_{2}}\in G_{2}..., R_{n}^{i_{n}}\in G_{n}, S_{1}\in G_{1}/N_{1}, S_{2}\in G_{2}/N_{2},...,S_{n}\in G_{n}/N_{n}$ where $R_{1}^{i_{1}}, R_{2}^{i_{2}},...,R_{n}^{i_{n}}$ are elements of cosets $R_{1}N_{1},R_{2}N_{2},..., R_{n}N_{n}$ and they are mapped by a homomorphism into elements $S_{1},S_{2},...,S{n}$, respectively, then $\Gamma$ contains $p=p_{1}p_{2}...p_{n}$ times the representation $\Gamma^{'}$.
\end{lemma}

\begin{proof}
The number of times that a given irreducible representation $\Gamma^{(\mu)}$ is contained in a representation $\Gamma^{'}$ is given by:
\begin{eqnarray}
a_{\mu}^{'}&=&\frac{n_{1}n_{2}...n_{n}}{g_{1}g_{2}...g_{n}}\sum_{S_{1}, S_{2},...,S_{n}}\chi^{(\mu_{1})^{\ast}}(S_{1})\chi^{(\mu_{2})^{\ast}}(S_{2})...\chi^{(\mu_{n})^{\ast}}(S_{n})
\chi^{{\Gamma_{1}}^{'}}(S_{1})\chi^{{\Gamma_{2}}^{'}}(S_{2}) \nonumber \\
&&...\chi^{{\Gamma_{n}}^{'}}(S_{n}), \nonumber
\end{eqnarray}
where $\mu_{1},\mu_{2},...\mu_{n}$ are irreducible representations of groups $G_{1},G_{2},...,G_{n}$, respectively. Note that the tensor product of irreducible representations of groups $G_{1},G_{2},...,G_{n}$ corresponds to an irreducible representation of group $G_{1}\times G_{2}\times...\times G_{n}$.
As each coset $RN_{i}$ has $|N_{i}|=n_{i}$ elements, we have
\begin{eqnarray}
a_{\mu}^{'}&=&\frac{1}{g_{1}g_{2}...g_{n}p_{1}p_{2}...p_{n}}\sum_{S_{1}, S_{2},...,S_{n}}[\chi^{(\mu_{1})^{\ast}}(R_{1})\chi^{(\mu_{2})^{\ast}}(R_{2})...\chi^{(\mu_{n})^{\ast}}(R_{n})
\chi^{{\Gamma_{1}}^{'}}(R_{1}) \nonumber \\
&&\chi^{{\Gamma_{2}}^{'}}(R_{2})...\chi^{{\Gamma_{n}}^{'}}(R_{n})]. \nonumber
\end{eqnarray}

On the other hand, the number of times that $\Gamma$ contains $\Gamma^{\mu}$ is
\begin{eqnarray}
a_{\mu}&=&\frac{1}{g_{1}g_{2}...g_{n}}\sum_{R_{1}\times R_{2}\times...\times R_{n}}\chi^{(\mu)^{\ast}}(R_{1}\times R_{2}\times...\times R_{n})\chi^{\Gamma}(R_{1}\times R_{2}\times...\times R_{n}) \nonumber \\
&=&\frac{1}{g_{1}g_{2}...g_{n}}\sum_{R_{1}\times R_{2}\times...\times R_{n}}[\chi^{(\mu_{1})^{\ast}}(R_{1})\chi^{(\mu_{2})^{\ast}}(R_{2})...\chi^{(\mu_{n})^{\ast}}(R_{n})
\chi^{\Gamma_{1}}(R_{1})\chi^{\Gamma_{2}}(R_{2}) \nonumber \\
&&...\chi^{\Gamma_{n}}(R_{n})]. \nonumber
\end{eqnarray}
Therefore
\begin{equation}
a_{\mu}=p_{1}p_{2}...p_{n}a_{\mu}^{'}.
\end{equation}
Thus $\Gamma$ contains $p$ times a representation $\Gamma^{'}$.

\end{proof}

A first application in quantum chemistry can be obtained from the next theorem.

\begin{theorem}
Let be $G=G_{1}\times G_{2}\times...\times G_{n}=(N_{1}\times N_{2}\times...\times N_{n})\rtimes (H_{1}\times H_{2}\times...\times H_{n})$ a finite group. Then every representation of $H_{1}\times H_{2}\times...\times H_{n}$ obtained from a irreducible representation of $G$ contains at most once a given irreducible representation of $H_{1}\times H_{2}\times...\times H_{n}$.
\end{theorem}

\begin{proof}
The number of times that a given representation $\Gamma=\Gamma_{1}\otimes \Gamma_{2}\otimes...\otimes\Gamma_{n}$ of $H=H_{1}\times H_{2}\times...\times H_{n}$ which also corresponds to an unique irreducible representation of $G$ contains a given irreducible representation $\nu=\nu_{1}\otimes \nu_{2}\otimes...\otimes\nu_{n}$ of the $H$ and $G$ is given by

\begin{eqnarray}
a_{\nu}&=&\frac{n_{1}n_{2}...n_{n}}{g_{1}g_{2}...g_{n}}\sum_{S_{1}\times S_{2}\times...\times S_{n}}[\chi^{(\nu_{1}\otimes \nu_{2}\otimes...\otimes\nu_{n})^{\ast}}(S_{1}\times S_{2}\times...\times S_{n}) \nonumber \\
&&\chi^{\Gamma}(S_{1}\times S_{2}\times...\times S_{n})], \nonumber
\end{eqnarray}
as $H=H_{1}\times H_{2}\times...\times H_{n}\simeq (G_{1}\times G_{2}\times...\times G_{n})/(N_{1}\times N_{2}\times...\times N_{n})$. By lemma 1, every irreducible representation of $G$ corresponds to an unique the representation of $H$, which is an irreducible representation of $H$. Indeed if we consider the inverse of lemma 1 with the restriction that the irreducible representation of   $G_{1}\times G_{2}\times...\times G_{n}$ also corresponds to an unique representation of $(G_{1}\times G_{2}\times...\times G_{n})/(N_{1}\times N_{2}\times...\times N_{n})$, we have:
\begin{eqnarray}
&&g_{1}g_{2}...g_{n}\delta_{\nu_{1}\nu_{2}...\nu_{n}, \Gamma_{1}\Gamma_{2}...\Gamma_{n}} \nonumber \\
&=&\sum_{R_{1}\times R_{2}\times...\times R_{n}}[\chi^{(\nu_{1}\otimes \nu_{2}\otimes...\otimes \nu_{n}) \ast }(R_{1}\times R_{2}\times...\times R_{n})\nonumber \\
&&\chi^{(\Gamma_{1}\otimes \Gamma_{2}\otimes...\otimes \Gamma_{n})}(R_{1}\times R_{2}\times...\times R_{n})] \nonumber \\
&=&n_{1}n_{2}...n_{n} \sum_{S_{1}\times S_{2}\times...\times S_{n}}[\chi^{(\nu_{1}\otimes \nu_{2}\otimes...\otimes \nu_{n}) \ast }(S_{1}\times S_{2}\times...\times S_{n}) \nonumber \\
&&\chi^{(\Gamma_{1}\otimes \Gamma_{2}\otimes...\otimes \Gamma_{n})}(S_{1}\times S_{2}\times...\times S_{n})].
\end{eqnarray}
Thus
\begin{eqnarray}
&&\sum_{S_{1}\times S_{2}\times...\times S_{n}}[\chi^{(\nu_{1}\otimes \nu_{2}\otimes...\otimes \nu_{n}) \ast }(S_{1}\times S_{2}\times...\times S_{n}) \nonumber \\
&&\chi^{(\Gamma_{1}\otimes \Gamma_{2}\otimes...\otimes \Gamma_{n})}(S_{1}\times S_{2}\times...\times S_{n})] \nonumber \\
&=&\frac{g_{1}g_{2}...g_{n}}{n_{1}n_{2}...n_{n}}\delta_{\nu_{1}\nu_{2}...\nu_{n}, \Gamma_{1}\Gamma_{2}...\Gamma_{n}}.
\end{eqnarray}
Therefore:
\begin{equation}
a_{\nu}=\frac{n_{1}n_{2}...n_{n}}{g_{1}g_{2}...g_{n}}\delta_{\nu_{1}\Gamma_{1}}\delta_{\nu_{2}\Gamma_{2}}...\delta_{\nu_{n}
\Gamma_{n}}h_{1}h_{2}...h_{n},
\end{equation}
by using the orthogonality relations for the characters of irreducible representations. Thus $a_{\nu}=0$ or $a_{\nu}=1$
\end{proof}

This theorem has an immediate physical implications. It indicates in perturbation theory that if the symmetry group of perturbed Hamiltonian is a subgroup $H$ of the group $G$ related to unperturbed Hamiltonian, with $G=N\rtimes H$, the energy levels do not dived itself, since $a_{\nu}=0$ or $a_{\nu}=1$. The following theorem is useful for the determination of the electronic transitions.

\begin{theorem}
Suppose $G=G_{1}\times G_{2}\times...\times G_{n}=(N_{1}\times N_{2}\times...\times N_{n})\rtimes (H_{1}\times H_{2}\times...\times H_{n})$ a finite group, $\Gamma^{(\varrho)}$ a reducible representation of G which corresponds to an unique representations of $G/N=(G_{1}\times G_{2}\times...\times G_{n})/(N_{1}\times N_{2}\times...\times N_{n})$ and $\Gamma^{(\mu)}$ and $\Gamma^{(\nu)}$ irreducible representations of $G$ and $G/N$, respectively. Then $\Gamma^{(\rho)}$ corresponds to an unique reducible representation of $G/N$ and the number of times that $\Gamma^{(\rho)}\otimes \Gamma^{(\mu)}$ contains the representation $\Gamma^{(\nu)}$ is given by:

\begin{eqnarray}
a_{\nu}&=&\frac{1}{h}\sum_{S_{1}, S_{2},..., S_{n}}[\chi^{(\mu_{1})}(S_{1})\chi^{(\rho_{1})}(S_{1})\chi^{(\nu_{1})^{\ast}}(S_{1}) \chi^{(\mu_{2})}(S_{2})\chi^{(\rho_{2})}(S_{2})\chi^{(\nu_{2})^{\ast}}(S_{2})\nonumber \\
&&...\chi^{(\mu_{n})}(S_{n})\chi^{(\rho_{n})}(S_{n})\chi^{(\nu_{n})^{\ast}}(S_{n})], \nonumber
\end{eqnarray}
where $h=|H|=h_{1}h_{2}...h_{n}$, $S_{i}$ is an element of the subgroup $H_{i}$ and $\chi^{(\mu_{i})}, \chi^{(\rho_{i})}$,
$\chi^{(\nu_{i})}$ are characters associated to representations $\Gamma^{(\mu_{i})}, \Gamma^{(\rho_{i})}, \Gamma^{(\nu_{i})}$ , respectively, of the groups $G_{i}$ and $G_{i}/N_{i}$. In particular if $\Gamma^{(\rho)}$ is a representations that satisfies the conditions of the lemma 2, the number of times $a_{\nu}^{'}$ that  $\Gamma^{(\rho)}\otimes \Gamma^{(\mu)}$ contains the representation $\Gamma^{(\nu)}$ is given by $a_{\nu}^{'}=pa_{\nu}$.
\end{theorem}
\begin{proof}
Since $\Gamma^{(\rho)}$ is a representation of $(G_{1}\times G_{2}\times...\times G_{n})/(N_{1}\times N_{2}\times...\times N_{n})$ and G is a semidirect product of groups, then $\Gamma^{(\rho)}$ corresponds to an unique representation of $H_{1}\times H_{2}\times...H_{n}$, as $H_{1}\times H_{2}\times...H_{n}\simeq (G_{1}\times G_{2}\times...\times G_{n})/(N_{1}\times N_{2}\times...\times N_{n})$. Besides, $\Gamma^{(\rho)}$ corresponds to an unique reducible representation of the $G/N$, because if $\Gamma^{(\rho)}$ were irreducible representation, we would have a contradiction by lemma 1. Similarly, $\Gamma^{(\mu)}$ and $\Gamma^{(\nu)}$ are also irreducible representations of the $(G_{1}\times G_{2}\times...\times G_{n})/(N_{1}\times N_{2}\times...\times N_{n})$. Furthermore the tensor product of irreducible representations is also an irreducible representation. Thus
\begin{eqnarray}
a_{\nu}&=&\frac{1}{h_{1}h_{2}...h_{n}}\sum_{S_{1}\times S_{2}\times...\times S_{n}}[\chi^{(\mu_{1}\otimes \mu_{2}\otimes...\otimes \mu_{n})}(S_{1}\times S_{2}\times...\times S_{n}) \nonumber \\
&&\chi^{(\rho_{1}\otimes \rho_{2}\otimes...\otimes \rho_{n})}(S_{1}\times S_{2}\times...\times S_{n})
\chi^{(\nu_{1}\otimes \nu_{2}\otimes...\otimes \nu_{n})^{\ast}}(S_{1}\times S_{2}\times...\times S_{n})] \nonumber \\
&=&\frac{1}{h}\sum_{S_{1},S_{2},...,S_{n}}[\chi^{(\mu_{1})}(S_{1})\chi^{(\rho_{1})}(S_{1})\chi^{(\nu_{1})^{\ast}}(S_{1}) \chi^{(\mu_{2})}(S_{2})\chi^{(\rho_{2})}(S_{2})\chi^{(\nu_{2})^{\ast}}(S_{2})\nonumber \\
&&...\chi^{(\mu_{n})}(S_{n})\chi^{(\rho_{n})}(S_{n})\chi^{(\nu_{n})^{\ast}}(S_{n})]. \nonumber
\end{eqnarray}
If $\Gamma^{(\rho)}$ satisfies the conditions of the lemma 2,
\begin{eqnarray}
a_{\nu}^{'}&=&\frac{1}{h}\sum_{S_{1}, S_{2},..., S_{n}}[\chi^{(\mu_{1})}(S_{1})p_{1}\chi^{(\rho_{1})}(S_{1})\chi^{(\nu_{1})^{\ast}}(S_{1}) \chi^{(\mu_{2})}(S_{2})p_{2}\chi^{(\rho_{2})}(S_{2})\chi^{(\nu_{2})^{\ast}}(S_{2})\nonumber \\
&&...\chi^{(\mu_{n})}(S_{n})p_{n}\chi^{(\rho_{n})}(S_{n})\chi^{(\nu_{n})^{\ast}}(S_{n})] \nonumber \\
&=&p_{1}p_{2}...p_{n}a_{\nu}.
\end{eqnarray}
Hence the theorem is proved.
\end{proof}

This theorem can be useful for the selection rules. If $\theta^{(\rho)}_{k}, \psi^{(\mu)}_{i}$ and  $\phi^{(\nu)}_{j}$ are basis function for the representations $\Gamma^{(\rho)}, \Gamma^{(\mu)}$ and $\Gamma^{(\nu)}$, respectively, we can determine whether the products or linear combinations of these products belong to $\nu$-th representation through characters of the a subgroup $H_{1}\times H_{2}\times...\times H_{n}$ of $G$. If $a_{\nu}\neq 0$ transitions occurs between energy levels $\mu$ and $\nu$.
We now show how to obtain projection operators in this context according to the Van Vleck procedure with the basis function generating machine \cite{Bishop}.

\begin{theorem}
Let $\Gamma^{(\rho)}=\Gamma^{(\rho_{1})}\otimes \Gamma^{(\rho_{2})}\otimes...\otimes\Gamma^{(\rho_{n})}$ be an irreducible representation of $(G_{1}\times G_{2}\times...\times G_{n})/(N_{1}\times N_{2}\times...\times N_{n})$ of dimension $d_{\rho}$, $G=G_{1}\times G_{2}\times...\times G_{n}=(N_{1}\rtimes H_{1})\times (N_{2}\rtimes H_{2})\times... (N_{n}\rtimes H_{n})$ a finite group that has $m$ non-equivalent irreducible representations and $O_{R_{i}}$ an operator associated to element $R_{i}\in G_{i}$. Thus \\
(i) the action of the operator
\begin{equation}
P_{\lambda_{1}\lambda_{2}...\lambda_{n}k_{1}k_{2}...k_{n}}^{(\rho)}=\frac{d_{\rho_{1}}d_{\rho_{2}}...d_{\rho_{n}}}{h_{1}h_{2}
...h_{n}}\sum_{(R_{1}R_{2}...R_{n})\in H\leq G}\Gamma_{\lambda_{1}k_{1}}^{\ast (\rho_{1})}\Gamma_{\lambda_{2}k_{2}}^{\ast (\rho_{2})}...\Gamma_{\lambda_{n}k_{n}}^{\ast (\rho_{n})}O_{R_{1}}O_{R_{2}}...O_{R_{n}}
\end{equation}
on a basis-function of the representation space results in $0$ unless it belongs to the $k_{1}k_{2}...k_{n}$-th row of the irreducible representation $\Gamma^{(\rho)}$ which also corresponds to an unique irreducible representation of $H_{1}\times H_{2}\times...\times H_{n}$, where $|H_{i}|=h_{i}$.\\
(ii) the operator $P_{k_{1}k_{2}...k_{n}k_{1}k_{2}...k_{n}}^{(\rho)}$  projects only part of the function
\begin{equation}
\Phi=\sum_{\rho_{1}=1}^{m_{1}}\sum_{\rho_{2}=1}^{m_{2}}...\sum_{\rho_{n}=1}^{m_{n}}
\sum_{k_{1}=1}^{d_{\rho_{1}}}\sum_{k_{2}=1}^{d_{\rho_{2}}}...\sum_{k_{n}=1}^{d_{\rho_{n}}}
\phi_{k_{1}}^{(\rho_{1})}\phi_{k_{2}}^{(\rho_{2})}...
\\\phi_{k_{n}}^{(\rho_{n})}
\end{equation}
that belongs to the $k_{1}k_{2}...k_{n}$-th row of the irreducible representation $\Gamma^{(\rho)}$. \\
(iii) the operator
\begin{eqnarray}
P^{(\rho)}&=&\sum_{k_{1}}\sum_{k_{2}}...\sum_{k_{n}}P_{k_{1}k_{2}...k_{n}k_{1}k_{2}...k_{n}}^{(\rho)} \nonumber \\
&=&\frac{d_{\rho_{1}}d_{\rho_{2}}...d_{\rho_{n}}}{h_{1}h_{2}
...h_{n}}\sum_{R_{1}}\sum_{R_{2}}...\sum_{R_{n}}\chi^{\ast (\rho_{1})}(R_{1})\chi^{\ast (\rho_{2})}(R_{2})\nonumber \\
&&...\chi^{\ast (\rho_{n})}(R_{n})O_{R_{1}\times R_{2} \times...\times R_{n}} \nonumber
\end{eqnarray}
projects the arbitrary function $\Phi$ defined above in a function $\phi^{(\rho)}$ belonging to the  $\rho$-th irreducible representation of the groups $G_{1}\times G_{2}\times...\times G_{n}$ and $H_{1}\times H_{2}\times...\times H_{n}$

\end{theorem}

\begin{proof}
For the item (i), we have that operator
\begin{equation}
P_{\lambda k}^{\rho}=\frac{d_{\rho}}{h}\sum_{R}\Gamma_{\lambda k}^{\ast(\rho)}(R)O_{R}
\end{equation}
can be writeen as
\begin{eqnarray}
P_{\lambda_{1}\lambda_{2}...\lambda_{n}k_{1}k_{2}...k_{n}}^{(\rho_{1}\otimes \rho_{2}\otimes...\rho_{n})}&=&\frac{d_{\rho_{1}}d_{\rho_{2}}...d_{\rho_{n}}}{h_{1}h_{2}...h_{n}}
\sum_{R_{1}\times R_{2}\times...\times R_{n}}[\Gamma_{\lambda_{1}\lambda_{2}...\lambda_{n}k_{1}k_{2}...k_{n}}^{\ast (\rho_{1}\otimes \rho_{2}\otimes...\rho_{n})}(R_{1}\times R_{2} \nonumber \\
&&\times...\times R_{n})O_{R_{1}\times R_{2}\times...\times R_{n}}] \nonumber \\
&=&\frac{d_{\rho_{1}}d_{\rho_{2}}...d_{\rho_{n}}}{h_{1}h_{2}...h_{n}}\sum_{R_{1}}\sum_{R_{2}}...\sum_{R_{n}}
[\Gamma_{\lambda_{1}k_{1}}^{\ast(\rho_{1})}(R_{1})\Gamma_{\lambda_{2}k_{2}}^{\ast(\rho_{2})}(R_{2}) \nonumber \\
&&...\Gamma_{\lambda_{n}k_{n}}^{\ast(\rho_{n})}(R_{n})O_{R_{1}R_{2}...R_{n}}]. \nonumber
\end{eqnarray}
By lemma 1 we have that $\Gamma^{(\rho)}$ corresponds to an unique irreducible representation of $G_{1}\times G_{2}\times...\times G_{n}$ and $H_{1}\times H_{2}\times...\times H_{n}$. Then we can use the orthogonality theorem for the irreducible representations of the $H_{1}\times H_{2}\times...\times H_{n}$, resulting in
\begin{eqnarray} \label{relation}
P_{\lambda_{1}\lambda_{2}...\lambda_{n}k_{1}k_{2}...k_{n}}^{(\rho_{1}\otimes \rho_{2}\otimes...\rho_{n})}
\varphi_{l_{1}}^{(i_{1})}\varphi_{l_{2}}^{(i_{2})}...\varphi_{l_{n}}^{(i_{n})}&=&
P_{\lambda_{1}}^{(\rho_{1})}\varphi_{l_{1}}^{(i_{1})}P_{\lambda_{2}}^{(\rho_{2})}\varphi_{l_{2}}^{(i_{2})}...
P_{\lambda_{n}}^{(\rho_{n})}\varphi_{l_{n}}^{(i_{n})} \nonumber \\
&=&\varphi_{l_{1}}^{(\rho_{1})}\delta_{i_{1}\rho_{1}}\delta_{k_{1}l_{1}}\varphi_{l_{2}}^{(\rho_{2})}
\delta_{i_{2}\rho_{2}}\delta_{k_{2}l_{2}}...\varphi_{l_{n}}^{(\rho_{n})}\delta_{i_{n}\rho_{n}}\delta_{k_{n}l_{n}} \nonumber \\
&=&\varphi_{l_{1}l_{2}...l_{n}}^{(i_{1}i_{2}...i_{n})}\delta_{i_{1}i_{2}...i_{n}j_{1}j_{2}...j_{n}}\delta_{k_{1}k_{2}...k_{n}
l_{1}l_{2}...l_{n}}.
\end{eqnarray}
For the item (ii), taking $\lambda_{1}=k_{1},\lambda_{2}=k_{2},...,\lambda_{n}=k_{n}$ in the eq. (\ref{relation}), we have
\begin{eqnarray}
P_{k_{1}k_{2}...k_{n}k_{1}k_{2}...k_{n}}^{(\rho_{1}\otimes \rho_{2}\otimes...\rho_{n})}\varphi_{l_{1}}^{(i_{1})}\varphi_{l_{2}}^{(i_{2})}...
\varphi_{l_{n}}^{(i_{n})}&=&P_{k_{1}}^{(\rho_{1})}\varphi_{l_{1}}^{(i_{1})}P_{k_{2}}^{(\rho_{2})}\varphi_{l_{2}}^{(i_{2})}...
P_{k_{n}}^{(\rho_{n})}\varphi_{l_{n}}^{(i_{n})} \nonumber \\
&=&\varphi_{k_{1}}^{(\rho_{1})}\varphi_{k_{2}}^{(\rho_{2})}...\varphi_{k_{n}}^{(\rho_{n})}
\delta_{i_{1}\rho_{1}}\delta_{i_{2}\rho_{2}}...\delta_{i_{n}\rho_{n}} \nonumber \\
&&\delta_{k_{1}l_{1}}\delta_{k_{2}l_{2}}...\delta_{k_{n}l_{n}}. \nonumber
\end{eqnarray}
Therefore
\begin{eqnarray}
&&P_{k_{1}k_{2}...k_{n}k_{1}k_{2}...k_{n}}^{(\rho_{1}\otimes \rho_{2}\otimes...\rho_{n})}
\sum_{\rho_{1}=1}^{m_{1}}\sum_{\rho_{2}=1}^{m_{2}}...\sum_{\rho_{n}=1}^{m_{n}} \sum_{k_{1}=1}^{d_{\rho_{1}}}\sum_{k_{2}=1}^{d_{\rho_{2}}}...\sum_{k_{n}=1}^{d_{\rho_{m}}}
\phi_{k_{1}}^{(\rho_{1})}\phi_{k_{2}}^{(\rho_{2})}...
\phi_{k_{n}}^{(\rho_{n})} \nonumber \\
&=&\phi_{k_{1}}^{(\rho_{1})}\phi_{k_{2}}^{(\rho_{2})}...\phi_{k_{n}}^{(\rho_{n})}. \nonumber
\end{eqnarray}
In the item (iii), we have
\begin{eqnarray}
P^{(\rho)}\Phi&=&\sum_{k_{1}}\sum_{k_{2}}...\sum_{k_{n}}P_{k_{1}k_{1}}^{(\rho_{1})}P_{k_{2}k_{2}}^{(\rho_{2})}...
P_{k_{n}k_{n}}^{(\rho_{n})}\Phi  \nonumber \\
&=&\frac{d_{\rho_{1}}d_{\rho_{2}}...d_{\rho_{n}}}{h_{1}h_{2}...h_{n}}\sum_{R_{1}\in H_{1}\leq G_{1}}
\sum_{R_{2}\in H_{2}\leq G_{2}}...\sum_{R_{n}\in H_{n}\leq G_{n}}
\sum_{k_{1}k_{2}...k_{n}}
[\Gamma_{k_{1}k_{1}}^{\ast (\rho_{1})}(R_{1}) \nonumber \\
&&\Gamma_{k_{2}k_{2}}^{\ast (\rho_{2})}(R_{2})...\Gamma_{k_{n}k_{n}}^{\ast (\rho_{n})}
(R_{n})\Phi] \nonumber \\
&=&\frac{d_{\rho_{1}}}{h_{1}}\sum_{R_{1}}\chi^{\ast (\rho_{1})}(R_{1})\sum_{\rho_{1}=1}^{m_{1}}\sum_{k_{1}=1}^{d_{\rho_{1}}}\phi_{k_{1}}^{(\rho_{1})}O_{R_{1}}
\frac{d_{\rho_{2}}}{h_{2}}\sum_{R_{2}}\chi^{\ast (\rho_{2})}(R_{2}) \nonumber \\
&&\sum_{\rho_{2}=1}^{m_{2}}\sum_{k_{2}=1}^{d_{\rho_{2}}}\phi_{k_{2}}^{(\rho_{2})}O_{R_{2}}...\frac{d_{\rho_{n}}}{h_{n}}\sum_{R_{n}}\chi^{\ast (\rho_{n})}(R_{n})\sum_{\rho_{n}=1}^{m_{n}}\sum_{k_{n}=1}^{d_{\rho_{n}}}\phi_{k_{n}}^{(\rho_{n})}O_{R_{n}} \nonumber \\
&=&\phi^{(\rho_{1})}\phi^{(\rho_{2})}...\phi^{(\rho_{n})},
\end{eqnarray}
as claimed
\end{proof}

It is interesting to note that the item (i) provides a recipe for to generate all the patterns of a given function related to an irreducible representation of $G_{1}\times G_{2}\times...\times G_{n}$ and $H_{1}\times H_{2}\times...\times H_{n}$, simultaneously, in according the Van Vleck scheme. The item (ii) shows how it is possible to obtain from an arbitrary function an element that belongs to the $k_{1}k_{2}...k_{n}$-th row of the irreducible representation.
In order to show how this approach can be extended we consider the next lemma

\begin{lemma}
Let $G=G_{1}\times G_{2}\times...\times G_{n}=(N_{1}\times N_{2}\times...\times N_{n})\rtimes (H_{1}\times H_{2}\times...\times H_{n})$ be a finite group. Then every subgroup of $H_{1}\times H_{2}\times...\times H_{n}$ is of the form $(G_{1}^{'}\times G_{2}^{'}\times...\times G_{n}^{'})/(N_{1}\times N_{2}\times...\times N_{n})$ with $N_{1}\times N_{2}\times...\times N_{n}\leq G_{1}^{'}\times G_{2}^{'}\times...\times G_{n}^{'} \leq G_{1}\times G_{2}\times...\times G_{n}$.
\end{lemma}

\begin{proof}
As $G_{1}\times G_{2}\times...\times G_{n}$ is a semidirect product of groups, we have:
\begin{equation}
H_{1}\times H_{2}\times...\times H_{n}\simeq \frac{G_{1}\times G_{2}\times...\times G_{n}}{N_{1}\times N_{2}\times...\times N_{n}}.
\end{equation}
By using the correspondence theorem and first isomorphism theorem \cite{Rose} we have that every subgroup of the $(G_{1}\times G_{2}\times...\times G_{n})/(N_{1}\times N_{2}\times...\times N_{n})$ is of the form $(G_{1}^{'}\times G_{2}^{'}\times...\times G_{n}^{'})/(N_{1}\times N_{2}\times...\times N_{n})$ where $N_{1}\times N_{2}\times...\times N_{n}\leq G_{1}^{'}\times G_{2}^{'}\times...\times G_{n}^{'} \leq G_{1}\times G_{2}\times...\times G_{n}$, as desirable.
\end{proof}

With this result the following theorem can be proved.

\begin{theorem}
Every irreducible representation of the a subgroup of the $H=H_{1}\times H_{2}\times...\times H_{n}$ with $G=(N_{1}\times N_{2}\times...\times N_{n})\rtimes (H_{1}\times H_{2}\times...\times H_{n})$ corresponds to an unique irreducible representation of the a group $G_{1}^{'}\times G_{2}^{'}\times...\times G_{n}^{'}$ with $N_{1}\times N_{2}\times...\times N_{n}\leq G_{1}^{'}\times G_{2}^{'}\times...\times G_{n}^{'} \leq G_{1}\times G_{2}\times...\times G_{n}$.
\end{theorem}

\begin{proof}
As every irreducible representation of the $(G_{1}^{'}\times G_{2}^{'}\times...\times G_{n}^{'})/(N_{1}\times N_{2}\times...\times N_{n})$ corresponds to an unique irreducible representation of the $G_{1}^{'}\times G_{2}^{'}\times...\times G_{n}^{'}$ by using the lemma 1, the theorem follows of the lemma 3.
\end{proof}

This last theorem shown that the analysis developed previously can be continued for the subgroups of the subgroups of the semidirect product.

\section{Applications}
Functional roles played by structural symmetry in macromolecules has been investigated \cite{GoodSell}. Symmetry is essential for allosteric regulation according to model proposed by Monod \emph{et al} \cite{Monod}. The association between protomers in an oligomer may be such as to confer an element of symmetry on the molecule is the first assumption of the model. Another assumption is that the symmetry of each set of stereospecific receptors is the same as the symmetry of the molecule. Besides, when the protein goes from one state to another state, its molecular symmetry is conserved. It is shown that symmetrical oligomeric complexes with two or more identical subunits are formed by most of the soluble and membrane-bound proteins found in living cells and nearly all structural proteins are symmetrical polymers of hundred to millions of subunits \cite{GoodSell}. The glycolate oxidase enzyme is an example of such proteins. In this section we consider the dihedral group $D_{4}$ corresponding to symmetry of the glycolate oxidase enzyme. The group $G=D_{4}\times D_{4}$ is formed by elements

\begin{eqnarray}
G&=&\{E\times E, E\times C_{2}(x), C_{2}(x)\times E, C_{2}(x)\times C_{2}(x),  \nonumber \\
&&E\times C_{2}(y), C_{2}(y)\times E, C_{2}(y)\times C_{2}(y), E\times C_{2}(z) \nonumber \\
&&C_{2}(z)\times E, C_{2}(z)\times C_{2}(z), C_{2}(x)\times C_{2}(y), C_{2}(y)\times C_{2}(x) \nonumber \\
&&C_{2}(x)\times C_{2}(z), C_{2}(z)\times C_{2}(x), C_{2}(y)\times C_{2}(z), C_{2}(z)\times C_{2}(y)\}. \nonumber
\end{eqnarray}
Consider the following subgroups
\begin{eqnarray}
N_{1}\times N_{2} &=&\{E\times E,E\times C_{2}(x),C_{2}(x)\times E, C_{2}(x)\times C_{2}(x)\}, \nonumber \\
H_{1}\times H_{2} &=&\{E\times E,E\times C_{2}(y),C_{2}(y)\times E, C_{2}(y)\times C_{2}(y)\}. \nonumber
\end{eqnarray}
We have that
\begin{equation}
G=D_{4}\times D_{4}=(N_{1}\times N_{2})\rtimes (H_{1}\times H_{2})
\end{equation}
and therefore
\begin{equation}
\frac{D_{4}\times D_{4}}{N_{1}\times N_{2}} \simeq H_{1}\times H_{2}.
\end{equation}
The subgroup $H_{1}\times H_{2}$ has four non-equivalent irreducible representations $\Gamma^{(1)}, \Gamma^{(2)}, \Gamma^{(3)}, \Gamma^{(4)}$ with characters \\

\begin{tabular}{|c|c|c|c|c|}
  \hline
  & $E\times E$ & $E\times C_{2}(y)$ & $%
  C_{2}(y)\times E$ & $C_{2}(y)\times C_{2}(y)$ \\ \hline
  $\Gamma^{(1)}$ & $1$ & $1$ & $1$ & $1$ \\ \hline
  $\Gamma^{(2)}$ & $1$ & $-1$ & $1$ & $-1$ \\ \hline
  $\Gamma^{(3)}$ & $1$ & $1$ & $-1$ & $-1$ \\ \hline
  $\Gamma^{(4)}$ & $1$ & $-1$ & $-1$ & $1$ \\ \hline
\end{tabular}
\newline
\newline

We can note that the cosets
\begin{eqnarray}
C_{1}&=&\{E\times E,E\times C_{2}(x),C_{2}(x)\times E, C_{2}(x)\times C_{2}(x)\} \nonumber \\
C_{2}&=&\{E\times C_{2}(y), E\times C_{2}(z), C_{2}(x)\times C_{2}(y), C_{2}(x)\times C_{2}(z)\} \nonumber \\
C_{3}&=&\{C_{2}(y)\times E, C_{2}(y)\times C_{2}(x), C_{2}(z)\times E, C_{2}(z)\times C_{2}(x)\} \nonumber \\
C_{4}&=&\{C_{2}(y)\times C_{2}(y), C_{2}(y)\times C_{2}(z), C_{2}(z)\times C_{2}(y, C_{2}(z)\times C_{2}(z)\} \nonumber
\end{eqnarray}
are mapped by a homomorphism into elements $E\times E,E\times C_{2}(y),C_{2}(y)\times E, C_{2}(y)\times C_{2}(y)$, respectively. It is easy to verify that these irreducible representations also corresponds to an unique irreducible representations of the $D_{4}\times D_{4}$.

\section{Conclusions}
The concept of inner semidirect product can be explored in the context of the representations of finite groups and it is possible to derive results with a direct physical interpretation in quantum chemistry. In this approach, we obtained some general theorems using the orthogonality relations, the isomorphism theorems and the correspondence theorem. In the scenario of the perturbation theory, we show that if the symmetry group of the Hamiltonian is a subgroup of the unperturbed Hamiltonian, the energy levels do not dived itself. Conditions for transitions between energy levels were determined as well as projection operators and basis functions. Importantly, our results are general and apply to composite systems since we consider the direct product of groups and therefore the tensor product of representations. Finally, we present an example using the dihedral group corresponding to symmetry of the glycolate oxidase enzyme. As perspectives, this study can be conducted using the outer semidirect product groups and its implications in quantum chemistry may be analyzed.

\subsection*{\sffamily \large Acknowledgments}
The author would like to thank Eric Pinto for discussions.


\end{document}